%% file: Paper.tex
\documentclass[12pt,a4paper]{article}
\renewcommand{\arraystretch}{1.5}
\usepackage{amssymb,amsbsy,amsfonts,amsmath}
\usepackage{color,comment,epsfig,graphicx,harvard,rotating}
\usepackage{adjustbox,amsthm,mathtools,bbm}
\usepackage{dsfont,hyperref,pifont,ulem}
\usepackage{subcaption}
\usepackage{settings}
\usepackage{tikz}
\usetikzlibrary{calc,arrows.meta,intersections,patterns,positioning,shapes.misc,fadings,through,decorations.pathreplacing}
\usepackage[dvipsnames]{xcolor}
\definecolor{myblue}{rgb}{0, 0., 0.51}
\definecolor{ColorOne}{named}{MidnightBlue}
\definecolor{ColorTwo}{named}{Dandelion}
\definecolor{ColorThree}{named}{Plum}
\usepackage[shortlabels]{enumitem}
\setlist[enumerate]{leftmargin=0.8cm}
\usepackage{tcolorbox}
\usepackage{booktabs}

\usepackage{threeparttable}
\usepackage{soul}
\usepackage{caption}
\usepackage{array}
\usepackage{pdflscape}

\citationstyle{dcu}

\begin{document}
\renewcommand{\arraystretch}{2}
\lineskip=1ex \baselineskip 5ex
\pagestyle{empty}
\begin{center}
\noindent{\Large\bf Treatment-effect heterogeneity and interactive fixed effects: Can we control for too much?}\vskip 3em

\begin{tabular}{crcrc}
\textbf{Murilo Cardoso} && \textbf{Bruno Ferman} && \textbf{Marcelo Fernandes}\\
\multicolumn{5}{c}{Sao Paulo School of Economics, FGV}
\end{tabular}\end{center}\vskip 3em

\noindent\textbf{Abstract:}~~This paper studies the interactive fixed effects (IFE) estimator in a panel-data setting with heterogeneous treatment effects. We show that, if the treatment-effect heterogeneity admits a linear factor structure, the IFE estimator could fail to recover the average treatment effect on the treated units. The problem arises because the interactive fixed effects absorb the heterogeneity in the treatment effect, creating a \textit{bad-control} problem. With time-invariant factors or unit-invariant loadings in the treatment effect heterogeneity, identification may further break down due to multicollinearity. These problems are not present in alternative estimation methods that exclude treated units in post-treatment periods from the factor estimation.

\vfill

{\footnotesize\noindent\textbf{Acknowledgments:}~~The authors acknowledge financial support from CAPES, FAPESP (2023/01728-0) and CNPq (302153/2022-5). The usual disclaimer applies.}

\newpage\pagestyle{plain}
\section{Introduction}

There are many panel-regression methods that allow for time-varying unobserved heterogeneity in the estimation of average treatment effects on the treated units (ATT). A prominent class of such methods replaces parallel-trends assumptions with the restriction that potential outcomes follow a linear factor structure in the absence of treatment \cite{arkhangelsky2024causalmodelslongitudinalpanel}. Among the many alternatives in the literature, a natural one is the interactive fixed effects (IFE) estimator that \cn{bai2009panel} introduces in his seminal paper. However, it is not \textit{a priori} clear whether the IFE estimator indeed recovers the ATT under heterogeneous treatment effects.

There is a vast literature that investigates heterogeneous treatment effects in difference-in-differences (DiD) settings under parallel-trends assumptions. It shows that two-way fixed effects (TWFE) estimators could well recover nonconvex or negatively weighted averages of treatment effects in the presence of treatment-timing variation \cite{de2020two,callaway2021difference,goodman2021difference,sun2021estimating,borusyak2024revisiting}. In contrast, we know very little about the behavior of the estimators that rely on a factor structure for the potential outcomes when untreated under heterogeneous treatment effects. 

We fill this gap by studying the behavior of the IFE estimator in the presence of heterogeneous treatment effects. We show that, in empirically relevant settings, heterogeneity in treatment effects can induce a fundamental identification problem, so that the IFE estimator fails to recover the ATT or any interpretable average of treatment effects. The problem arises if the heterogeneous treatment effects admit a linear factor structure, as in the case that treatment effects depend on common latent shocks that affect all treated units with different intensity levels. For instance, the effect of a job-search assistance program may depend on aggregate unemployment conditions common to all individuals, while personal characteristics determine the magnitude of the response.

In this setting, the interaction between the treatment indicator and heterogeneous effects admits a linear factor structure. As a result, when using post-treatment outcomes of treated units to estimate factors as in the IFE estimator, part of the treatment effect might end up in the factor estimates. If the IFE estimator partials out variation due to the treatment, it would then fail to identify the ATT. There is a \textit{bad control} problem in that the IFE estimator cannot distinguish between latent factors governing potential outcomes when untreated (that it should control) and latent factors stemming from the heterogeneity in treatment effects (that it should not). As such, it absorbs both components, precluding their individual identification.

The issue becomes more severe if the heterogeneous treatment effect contains either time-invariant factors or unit-invariant loadings. In this case, the interaction between treatment status and the constant factor or loading is perfectly collinear with the treatment indicator, violating the rank conditions for identification in \cn{bai2009panel}. As a result, the IFE estimator is not necessarily consistent for the ATT and, in addition, it could even fail to converge to a unique probability limit.

\cn{xu2017generalized} argues that treatment-effect heterogeneity could well generate bias in the IFE estimator. However, he does not formally characterize when and why heterogeneity causes problems in such a setting. In contrast, we fully establish the conditions under which heterogeneity in treatment effects leads to inconsistency, thereby offering some intuition for the particular forms of heterogeneity that drive these failures. Moreover, we revisit the data-generating process that \cn{xu2017generalized} considers in his Monte Carlo simulations to illustrate concerns about heterogeneous treatment effects. We show that the IFE estimator is inconsistent under a static specification, but consistent under a dynamic specification. We also demonstrate that we recover consistency when the proportion of post-treatment periods or the proportion of treated units becomes negligible. This illustrates that heterogeneity in treatment effects does not necessarily imply inconsistent ATT estimates. Rather, whether inconsistency arises depends on several features, including the form of treatment-effect heterogeneity, the specification used for estimation, and the proportion of treated units and post-treatment periods.
 
The problems we identify are specific to estimation methods that use post-treatment outcomes of treated units to estimate the factor structure, such as the \pc{bai2009panel} IFE estimator. Alternative approaches, such as synthetic-control and imputation-based methods, do not suffer from this \textit{bad-control} problem, precisely because they exclude treated units in post-treatment periods from factor estimation. See, among others, \cn{abadie2010synthetic}, \cn{Doudchenko}, \cn{gobillon2016regional}, \cn{xu2017generalized}, \cn{amjad2018robust}, \cn{arkhangelsky2021synthetic}, \cn{athey2021matrix}, \cn{ben2021augmented}, \cn{ferman2021synthetic}, and \cn{porreca2022synthetic}.

Interestingly, the issue we raise here differs markedly from those affecting the TWFE estimator under heterogeneous treatment effects. In the latter case, the problem is that the estimator implicitly relies on comparisons between units at different stages of treatment exposure, including comparisons between units that are newly treated and units that remain treated over the same period. If treatment effects are heterogeneous, such comparisons could induce negative weighting of treatment effects in the presence of variation in treatment timing. In stark contrast, our point about the IFE estimator derives from a \textit{bad-control} issue. By using post-treatment outcomes of treated units to estimate the factor structure, the IFE estimator could eventually partial out components of the treatment effect itself even in settings without staggered treatment adoption.

We organize the rest of this paper as follows. Section~\ref{sec:Setting} introduces the framework. Section~\ref{sec:IFE} derives the main theoretical results. In particular, we first review the IFE estimator under homogeneous treatment effects in Section~\ref{sec:Homogenous} and then examine the case of heterogeneous treatment effects and the resulting \textit{bad-control} and collinearity issues in Section~\ref{sec:heterogenous}. Section~\ref{sec:MC.bad.control} reports some Monte Carlo simulations to illustrate our point, whereas Section~\ref{sec:Empirics} provides an empirical illustration. Section~\ref{sec:Conclusion} concludes.

\section{Setting} \label{sec:Setting}

Consider a panel-data setting in which we observe $i\in\{1,\ldots,N\}$ units at $t\in\{1,\ldots,T\}$ periods. We want to estimate the effect of a policy $d_{i,t}=d_id_t$, where $d_i=\bs{1}(i\in\mathcal{T})$ with $\mathcal{T}$ indicating the set of treated units and $d_t=\bs{1}(t>T_0)$ denoting an irreversible treatment that starts immediately after period $T_0<T$. In addition, we denote the number of post-treatment periods by $T_1=T-T_0<T$. We then define the potential outcome (PO) depending on whether unit $i$ is treated or not at time $t$: $y_{i,t}(1)$ and $y_{i,t}(0)$, respectively. In particular, we assume that the potential outcomes for untreated units admit a linear factor representation:
\begin{equation}\label{eq:potential.outcomes}
    y_{i,t}(0)=\bs{\lambda}_{0,i}'\bs{F}_{0,t}+e_{i,t}~~\mbox{and}~~y_{i,t}(1)=\alpha_{i,t}+y_{i,t}(0),
\end{equation}
where $\bs{F}_{0,t}$ and $\bs{\lambda}_{0,i}$ are $k_0\times 1$ vectors of deterministic factors and their corresponding loadings, $e_{i,t}$ is the idiosyncratic error for unit $i$ at time $t$, and $\alpha_{i,t}$ is the (possibly heterogeneous) treatment effect for unit $i$ at time $t$.

We consider a setting conditional on treatment assignment, so that we can treat the latter as fixed. We denote by $N_1$ and $N_0$ the number of treated and control units, respectively. We entertain only deterministic heterogeneous treatment effects, in that $\alpha_{i,t}$ is a fixed parameter, even if it potentially changes across units and over time \cite{alvarez2025inference}. The parameter of interest is the average treatment effect on treated units (ATT):
\begin{equation} \label{eq:alpha.bar}
    \bar\alpha=\frac{1}{N_1T_1}\sum_{i\in\mathcal{T}}\sum_{t=T_0+1}^T\alpha_{i,t}.
\end{equation}
The main challenge in the estimation of $\bar\alpha$ is that we observe realized (rather than potential) outcomes: namely, $y_{i,t}=(1-d_{i,t})\,y_{i,t}(0)+d_{i,t}\,y_{i,t}(1)$. We next summarize the sampling scheme and exogeneity conditions in Assumptions~\ref{A:sampling} and \ref{A:exogeneity}.

\begin{assumption}[Sampling]\label{A:sampling}
We observe the sample $\{y_{i,1},\ldots,y_{i,T}\}_{i=1}^N$, where we draw $y_{i,t}=(1-d_{i,t})\,y_{i,t}(0)+d_{i,t}\,y_{i,t}(1)$ from the PO model in \eqref{eq:potential.outcomes}, with $d_{i,t}$ taking value one if $i\in\mathcal{T}$ and $t>T_0$, zero otherwise. In addition, $\bs{\lambda}_{0,i}$, $\bs{F}_{0,t}$ and $\alpha_{i,t}$ are deterministic.
\end{assumption}
    
\begin{assumption}[Idiosyncratic errors] \label{A:exogeneity}
The idiosyncratic term $e_{i,t}$ is independent and identically distributed (iid) over time, with mean zero and finite fourth moment for all $i$.
\end{assumption}

While we consider a setting in which treatment allocation and factor structure are fixed, we could alternatively rewrite the framework as conditional on them. In this case, we would capture the characteristics of the treatment assignment mechanism by the distribution of potential outcomes conditional on treatment allocation. For example, if the units select into treatment based on unobservable (to the econometrician) factors, then the conditional distribution of these unobservables given treatment allocation depends on treatment status \ca{Ferman_JASA}{see discussion in}. In addition, we would have to formulate Assumption~\ref{A:exogeneity} in terms of conditional moments to ensure that treatment assignment is as good as random once we control for the factor structure. Finally, it is worth saying that we require iid errors just for convenience \ca{bai2009panel}{see, for instance, discussion in}.

The linear factor structure for the potential outcomes when untreated implies that the DiD estimator is not necessarily unbiased for $\bar\alpha$, except in special cases such as if $\bs{\lambda}_{0,i}'\bs{F}_{0,t}=\theta_i+\gamma_t$. The two-way fixed effects (TWFE) allows for confounders $\theta_i$ that are unit-specific but constant over time and time-varying confounders $\gamma_t$ that affect all units in the same way. This specification rules out the possibility of time-varying unobserved confounders that differ across units, though.

\section{ATT estimators} \label{sec:IFE}

In this section, we discuss the estimation of $\bar\alpha$ using the IFE estimator, as well as some alternative estimators that assume a linear factor model for the potential outcomes.

\subsection{Homogeneous treatment effects}\label{sec:Homogenous}

We first examine the identification and estimation of treatment effects using the IFE framework under the assumption that the treatment effect is homogeneous across units and periods: $\alpha_{i,t}=\alpha_0$ for all $i$ and $t$. Let now $\bs{Y}_i=(y_{i,1},\ldots,y_{i,T})'$, $\bs{D}_i=(d_{i,1},\ldots,d_{i,T})'$, $\bs{F}_0=(\bs{F}_{0,1}',\ldots,\bs{F}_{0,T}')'$, $\bs{\Lambda}_0=(\bs\lambda_{0,1}',\ldots,\bs\lambda_{0,N}')'$, and $\bs{e}_i=(e_{i,1},\ldots,e_{i,T})'$. We can then represent the PO model in \eqref{eq:potential.outcomes} as
\begin{equation}\label{eq:IFE.model}
    \bs{Y}_i=\alpha_0\bs{D}_i+\bs{F}_0\bs\lambda_{0,i}+\bs{e}_i.
\end{equation}

\cn{bai2009panel} defines the IFE estimator as
\begin{equation}\label{eq:SSE_min}
(\hat\alpha,\hat{\bs{F}})=\argmin_{(\alpha,\bs{F})}\mbox{SSE}(\alpha,\bs{F})=\sum_{i=1}^N(\bs{Y}_i-\alpha\bs{D}_i)'\bs{M}_{\bs{F}}(\bs{Y}_i-\alpha\bs{D}_i),
\end{equation}
where $\bs{M}_{\bs{F}}=\bs{I}_{T}-\bs{F}(\bs{F}'\bs{F})^{-1}\bs{F}'=\bs{I}_{T}-\bs{F}\bs{F}'/T$. It differs from the least-squares estimator only because it must also retrieve the latent factor structure. Implementation rests on an interactive procedure that estimates both the parameter of interest and the factor structure. \cn{bai2009panel} establishes consistency of $\hat\alpha$ for $\alpha_0$ as $N,T\rightarrow\infty$, under Assumptions~\ref{A:sampling} and~\ref{A:exogeneity}, and the regularity conditions~\ref{A:bai2009} in Appendix~\ref{sec:Proofs}. Although his result treats the dimension of the latent factors as known, \cn{moon2015linear} demonstrate that estimating the factor dimension does not affect the asymptotic behavior of $\hat\alpha$ as long as we do not underestimate the number of factors. 

There are alternative methods to estimate causal effects in the event that potential outcomes follow a linear factor structure. See, for instance, the synthetic-controls (SC) approach by \cn{abadie2010synthetic}, the generalized synthetic-controls (GSC) framework by \cn{xu2017generalized}, the demeaned synthetic-controls (DSC) estimator by \cn{Doudchenko} and \cn{ferman2021synthetic}, and the synthetic difference-in-differences (SDiD) approach by \cn{arkhangelsky2021synthetic}. In contrast to the IFE estimator, they do not rely on treated units in the post-treatment period to estimate the factor structure.

\subsection{Heterogeneous treatment effects} \label{sec:heterogenous}

We next relax the homogeneity assumption on the treatment effects by assuming they admit a linear factor representation: namely, $\alpha_{i,t}=\gamma+\bs\lambda_{\alpha,i}'\bs{F}_{\alpha,t}$, where $\gamma$ represents the common treatment effect, and $\bs{F}_{\alpha,t}$ and $\bs\lambda_{\alpha,i}$ denote the $k_\alpha-$vectors of latent factors and loadings. The main problem arises from the fact that the PO model under heterogeneous treatment effects admits a linear factor representation that combines both $\bs{F}_{0,t}$ and $\bs{F}_{\alpha,t}$: 
\begin{equation}\label{eq:IFE.bad.control}
\bs{Y}_i=\gamma\bs{D}_i+\bb{F}\bb{\lambda}_i+\bs{e}_i. 
\end{equation} 
This happens because $d_{i,t}\bs\lambda_{\alpha,i}'\bs{F}_{\alpha,t}+\bs\lambda_{0,i}'\bs{F}_{0,t}=(d_i\bs\lambda_{\alpha,i})' (d_t\bs{F}_{\alpha,t})+\bs\lambda_{0,i}'\bs{F}_{0,t}$, giving way to the extended factor representation  in~\eqref{eq:IFE.bad.control}, with $\bb F_t=(d_t\bs{F}_{\alpha,t}^\prime,\bs{F}_{0,t}^\prime)^\prime$ and $\bb\lambda_i=(d_{i}\bs\lambda_{\alpha,i}^\prime,\bs\lambda_{0,i}^\prime)^\prime$ such that $\bb\Lambda'\bb\Lambda$ is diagonal. We next impose some high-level conditions on this extended linear factor structure with factors $\bb F_t$ and loadings $\bb\lambda_i$.

\begin{assumption}[Extended Factor Structure]\label{A:heterogeneity.tv.factor}
The heterogeneity in the treatment effects follows the deterministic factor structure $\alpha_{i,t}=\gamma+\bs\lambda_{\alpha,i}'\bs{F}_{\alpha,t}$, where $\bs{F}_\alpha=(\bs{F}_{\alpha,1},\ldots,\bs{F}_{\alpha,T})'$ and $\bs\Lambda=(\bs\lambda_{\alpha,1},\ldots,\bs\lambda_{\alpha,N})'$. We assume that $\bb F_t=(d_t\bs{F}_{\alpha,t},\bs{F}_{0,t})$ and $\bb\lambda_i=(d_{i}\bs\lambda_{\alpha,i},\bs\lambda_{0,i})$ satisfy the following conditions:
\begin{enumerate}[(i)] 
\item $\bb F'\bb F/T=\bs I_{\breve{k}}$, with $\breve{k}\le k_0+k_\alpha$, and $\bb\Lambda'\bb\Lambda$ is a diagonal matrix. 
\item $\frac{1}{T}\sum_{t=1}^T\bb{F}_t'\bb{F}_t\to\bs{\Sigma}_{\bb F}$ as $T\to\infty$ and $\frac{1}{N}\sum_{i=1}^N \bb{\lambda}_i'\bb{\lambda}_i\to\bs{\Sigma}_{\bb \Lambda}$ as $N\to\infty$, where both $\bs\Sigma_{\bb F}$ and $\bs\Sigma_{\bb\Lambda}$ are positive definite matrices. Moreover, for a constant $M>0$, $||\bb{F}_{t}||\le M$ and  $||\bb{\lambda}_{i}||\le M$, where $||\bs A||=\sqrt{tr(\bs A'\bs A)}$ for any matrix $\bs A$.
\item ${\mathcal{\bb D}}(\bb F)=\frac{1}{NT}\sum_{i=1}^N\bs{D}_i'\bs{M}_{\bb F}\bs{D}_i-\frac{1}{N^2T}\sum_{i=1}^N\sum_{j=1}^N\bs{D}_i'\bs{M}_{\bb F}\bs{D}_j\bb\lambda_i'(\bb\Lambda'\bb\Lambda/N)^{-1}\bb\lambda_j>0$.
\end{enumerate} 
\end{assumption} 

Although we restrict attention to a deterministic factor structure, assuming a stochastic idiosyncratic component to the heterogeneity in treatment effects would not affect the results. It is just for simplicity that we do not entertain this case. Assumption~\ref{A:heterogeneity.tv.factor}$(i)$ is just a normalization, given that we can always find $\bb F$ and $\bb\Lambda$ that satisfy these conditions \cite{bai2009panel}. Note that $\breve{k}\le k_0+k_\alpha$ because the dimension of the linear subspace spanned by $\bb F$ is at most the sum of the dimensions of the linear subspaces spanned by $\bs F_0$ and $\bs F_\alpha$. For instance, if $Y(0)$ contains factors that are active only for treated units in post-treatment periods, some components of $(d_i\bs{\lambda}_{\alpha,i})'(d_t\bs{F}_{\alpha,t})$ could then lie in the factor space of $Y(0)$, thereby violating Assumption~\ref{A:heterogeneity.tv.factor}$(ii)$. We would have to drop the redundant factors so that the resulting extended factor structure $(\bb F,\bb\Lambda)$ would have a dimension smaller than $k_0+k_\alpha$. This does not affect the proof or the conclusions of Proposition~\ref{Prop:bad.control}, though.

Assumption~\ref{A:heterogeneity.tv.factor}$(ii)$ implies that factors and loadings in the extended factor representation asymptotically generate variation in the outcomes. This assumption would not hold, for example, if the proportion of treated units and/or periods becomes negligible as the sample size grows (see Remark~\ref{rmk:non-negligible}). Finally, Assumption~\ref{A:heterogeneity.tv.factor}$(iii)$ ensures that the extended factor structure is not collinear with the treatment dummy. We discuss in Remark~\ref{rmk:fixed.effect} two important cases in which this assumption could fail.

We next document that the minimization problem in \eqref{eq:SSE_min} does not result in a consistent estimator for $\bar\alpha$. 

\begin{proposition}[Bad Control]\label{Prop:bad.control}
Consider the PO model in \eqref{eq:potential.outcomes} and let Assumptions~\ref{A:sampling} to~\ref{A:heterogeneity.tv.factor} hold. It then follows that $\hat{\alpha}\xrightarrow{p}\gamma$ as $N,T\to\infty$.
\end{proposition}

Intuitively, the objective function increases with the heterogeneity among treated units, so the minimizer forces the factor structure to absorb it entirely. A direct implication of Proposition~\ref{Prop:bad.control} is that the IFE estimator $\hat\alpha$ is not a consistent estimator of the average treatment effect on the treated (ATT). A useful way to interpret this result is that heterogeneous treatment effects generate a factor structure that appears only for treated units in the post-treatment period. Because the IFE estimator searches for latent factors that explain systematic variation in outcomes, it attributes this pattern to additional factors rather than to treatment effects. As such, the factor structure we estimate absorbs part of the heterogeneous treatment effects, effectively acting as a bad control that removes this variation from the treatment effects. This explains why the IFE estimator converges to $\gamma$ rather than to the ATT parameter.

\begin{remark}[{Negligible treatments}]\label{rmk:non-negligible} Assumption~\ref{A:heterogeneity.tv.factor}(ii) rules out settings in which the proportion of post-treatment periods or the proportion of treated units are asymptotically negligible (i.e., $T_1/T\to 0$ or $N_1/N\to 0$, respectively). In such cases, the factor structure induced by heterogeneous treatment effects among treated units in post-treatment periods is asymptotically negligible. The IFE estimator does not absorb this systematic variation in outcomes through the factor structure, then. As a result, Proposition~\ref{Prop:bad.control} would no longer hold in view that the IFE estimator would not treat the heterogeneity in treatment effects as a bad control.
\end{remark}

\begin{remark}[Invariant factor structure]\label{rmk:fixed.effect} Assumption~\ref{A:heterogeneity.tv.factor}(iii) rules out either unit-invariant treatment effects (i.e., $\alpha_{i,t}=\psi_t$) or time-invariant factors (i.e., $\alpha_{i,t}=\phi_i$). In these cases, the interaction $d_{i,t}\alpha_{i,t}$ between treatment status and heterogeneous treatment effects would admit a linear factor structure that is collinear with the treatment indicator itself. As such, the objective function in \eqref{eq:SSE_min} might have multiple minima, so that the distribution of $\hat{\alpha}$ could well become multimodal in finite samples. Intuitively, once we partial out the factors, little variation remains in the treatment dummy (see simulations in Appendix~\ref{App:sim.time-invariant}).
\end{remark}

\begin{remark}[Dynamic IFE specification]\label{rmk:dynamic} To allow treatment effects to vary over time, one could think of interacting the treatment dummy with indicators for each post-treatment period. If the heterogeneity in treatment effects takes the form $\alpha_{i,t}=\psi_t$, we are essentially back to homogeneous treatment effects if we consider the dynamic specification $y_{i,t}=\sum_{s>T_0}\psi_s\bs{1}(t=s)d_{i,t}+\bs\lambda_{0,i}'\bs F_{0,t}+e_{i,t}$ and, as such, the IFE estimator is consistent for $(\psi_{T_0+1},\ldots,\psi_T)$. As before, adding a purely idiosyncratic term $a_{i,t}$ to the heterogeneity does not make any difference given that, by definition, idiosyncratic components do not admit a factor representation. Accordingly, under this dynamic specification, bad controls ensue only if the heterogeneity in treatment effects has a linear factor structure beyond $\psi_t$ (see simulations in Appendix~\ref{App:sim.dynamic}).
\end{remark}

\pc{xu2017generalized} simulations illustrate potential problems with the IFE estimator under heterogeneous treatment effects of the form $\alpha_{i,t}=\psi_t+a_{i,t}$, where $a_{i,t}$ is purely idiosyncratic. An interesting implication of Remark~\ref{rmk:dynamic} is that, while the static IFE estimator is indeed inconsistent under this data-generating process (DGP), the IFE estimator remains consistent in the dynamic specification. This highlights that whether heterogeneous treatment effects lead to inconsistency depends not only on the form of the heterogeneity, but also on the specification used for estimation.

\begin{remark}[Alternative causal estimators]\label{rmk:other estimators}
Estimation methods that do not use treated units in the post-treatment periods to estimate the factor structure are not subject to these problems. This is the case for the SC, GSC, DSC and SDiD estimators, which estimate the ATT as $\frac{1}{N_1T_1}\sum_{i\in\mathcal{T}}\sum_{t=T_0+1}^T\big(y_{i,t}-\hat y_{i,t}(0)\big)$, where $\hat y_{i,t}(0)$ is the corresponding counterfactual. Heterogeneity in treatment effects does not affect $\hat y_{i,t}(0)$ because it does not employ information on treated units in post-treatment periods. In our setting, $y_{i,t}-\hat y_{i,t}=\gamma+\bs{\lambda}_{\alpha,i}'\bs{F}_{\alpha,t} +y_{i,t}(0)-\hat y_{i,t}(0)$. As such, these methods are able to recover the ATT even in the presence of treatment effect heterogeneity, provided that we can accurately estimate the counterfactual outcomes for the treated units in the post-treatment periods.
\end{remark}

\section{Simulations}\label{sec:MC.bad.control}

To illustrate how factors affect the estimation of ATT under heterogeneous treatment effects, we carry out a Monte Carlo study using the PO model in \eqref{eq:potential.outcomes}, with \eqref{eq:alpha.bar} as the parameter of interest. We entertain two data generating processes. The first considers a homogeneous treatment effect in which $\alpha_{i,t}=2$ for all units and time periods, implying $\bar \alpha = 2$. The second features heterogeneous treatment effects $\alpha_{i,t}=\gamma+\bs\lambda_{\alpha,i}'\bs{F}_{\alpha,t}$, with $\gamma=1$ and $\frac{1}{N_1T_1}\sum_{i\in\mathcal{T}} \sum_{t=T_0+1}^T\bs\lambda_{\alpha,i}'\bs{F}_{\alpha,t}=1$, so that $\bar\alpha=2$. For more details on the data generating processes and the factor structure of the heterogeneous treatment effects, see the notes of Table~1 as well as Appendix~\ref{App:setting.sim}.

We consider five ATT estimators: IFE, SC, DSC, GSC, and SDiD. They are all consistent for the DGP with homogeneous treatment effects. Under heterogeneous treatment effects, however, the IFE estimator captures only the common treatment-effect component $\gamma$, while the PO factor structure soaks up $\bs\lambda_{\alpha,i}'\bs{F}_{\alpha,t}$ as a bad control (see Section~\ref{sec:heterogenous}). In contrast, this problem does not arise for the other estimators, since they do not use treated units in post-treatment periods to estimate the factor structure (see Remark~\ref{rmk:other estimators}).

Table~\ref{tab:simulation_bad_control} displays the simulation results. For the DGP with homogeneous treatment effects, all estimators are, on average, close to the ATT $\bar\alpha=2$. The IFE estimator performs slightly better in terms of variance relative to the other alternatives. This should not come as a surprise, since it uses more information: namely, the treated units in the post-treatment periods—to estimate the factor structure. Importantly, the assumption that treatment effects are homogeneous is crucial for the IFE estimator to recover only the factor structure of the untreated potential outcomes, rather than inadvertently controlling for part of the treatment effect.

Once we turn to heterogeneous treatment effects, the picture changes dramatically. The IFE estimator completely fails to recover the ATT parameter $\bar\alpha=2$, retrieving instead $\gamma=1$. Following the intuition from Section~\ref{sec:heterogenous}, the treatment-effect component $\bs\lambda_{\alpha,i}'\bs{F}_{\alpha,t}$ is now absorbed into the PO factor structure rather than into the treatment effect as it should. In contrast, the alternative estimators remain close to the true ATT value of $\bar\alpha=2$.

\section{Empirical illustration}\label{sec:Empirics}

To illustrate empirically the \textit{bad-control} issue in the IFE estimation, we revisit the study by \cn{gobillon2016regional} about the effect of the Enterprise Zone Program on the unemployment rate. In particular, we assess how the Enterprise Zone Program affects the exit rates from unemployment, either to a job or for unknown reasons, and the entry rate into unemployment (all on a logarithmic scale). The analysis is for the Paris region, with filters to match the sample in \cn{gobillon2016regional}, resulting in 135 control units and 13 treated units for 7 pre-treatment and 13 post-treatment periods.  

\input{Tables/Comparison}

For each outcome, we compare the IFE estimator of $\bar\alpha$ with the alternative SC, DSC, GSC and SDiD estimators. Although the latter approaches yield robust estimators to heterogeneous treatment effects, they rely on different assumptions for consistency. As such, they could well converge to different probability limits in the event some of these assumptions fail. For instance, the SC, DSC, and SDiD approaches impose convex weighting constraints on the factor loadings. In particular, the standard SC estimator is the most restrictive, calling for treated units' factor loadings lying in the convex hull of the donors’ loadings for both time-variant and time-invariant factors. Both DSC and SDiD approaches relax this assumption by previously removing time-invariant fixed effects, so that treated units must lie in the convex hull of the factor loadings, excluding those corresponding to time-invariant factors. In view that the GSC estimator does not require convex-hull assumptions, we take it as the main benchmark to assess whether the IFE estimator is indeed recovering the ATT parameter. 

\input{Tables/Empirical}

Table~\ref{tab:empirical} summarizes our empirical findings. We conduct inference by employing a permutation-based placebo test under the sharp null hypothesis of no treatment effect. This means that, at each permutation, we randomly select $N_1$ units as pseudo-treated before computing the ATT estimators and their corresponding test statistics. We denote the latter by $\theta_r^k$ for $k\in\{\text{IFE},\text{SC},\text{DSC}, \text{GSC},\text{SDiD}\}$ and $r\in\{1,\ldots,R\}$. We then compute the p-value of $\theta_r^k$ using the empirical distribution across $R=10{,}000$ permutations, i.e., $\frac{1}{R}\sum_{r=1}^R\mathbb{I}(\theta_r^k>\theta_{\bar{\alpha}}^k)$, where $\theta_{\bar{\alpha}}^k$ is the sample value of the test statistic using the observed data. For the SC and DSC estimators, we use the mean squared prediction error ratio as the test statistic to adjust for imperfect pre-treatment fit \cite{abadie2010synthetic,abadie2015comparative,firpo2018synthetic}. In turn, we consider the absolute treatment effect $|\hat{\alpha}_r^{k}|$ as the test statistic for the other estimators.

For the exit rate to a job (Column 1), the IFE estimator departs from all other estimators, indicating a significant positive ATT in contrast to the negative signs of the other ATT estimates. In particular, it differs sharply from the GSC, which suggests the presence of treatment effect heterogeneity. It is also interesting to observe that the GSC approach yields a much more prominent ATT than the SC, DSC and SDiD estimators, perhaps suggesting that the convex-hull conditions do not hold.

The latter does not seem a problem for the estimation of the ATT of the Enterprise Zone Program on the exit rate for unknown reasons in that most estimates are close to 5\% and mostly significant (Column 2). The only exception is the IFE approach, which yields an ATT estimate very close to zero. This pattern strongly suggests that the convex hull restrictions for the validity of the SC, DSC, and SDiD estimators hold, whereas the heterogeneous treatment effects prevent the IFE estimator from recovering the ATT. Finally, IFE and GSC estimators for entry rate are close to zero and not statistically significant at standard significance levels, which might suggest that treatment effects are homogeneous and equal to zero for this outcome (column 3).

\section{Conclusion}\label{sec:Conclusion}

We examine the behavior of \pc{bai2009panel} IFE estimator under heterogeneous treatment effects. We find that, if the heterogeneity in treatment effects admits a linear factor representation, the IFE estimator fails to recover the average treatment effect on the treated. This happens because the interactive fixed effects in the PO model absorb the factors in the heterogeneity, thereby becoming a bad control in that it accounts for too much. We also show that, in the presence of time-invariant components in the heterogeneous treatment effects, identification breaks down due to multicollinearity. In contrast, these issues do not arise for alternative estimators that are robust to heterogeneous treatment effects, essentially because they exclude treated units in post-treatment periods from the factor estimation in the PO model.

Empirically, we revisit \pc{gobillon2016regional} analysis about how the Enterprise Zone Program affects the unemployment rate in the Paris region. The IFE estimator yields very different results relative to the generalized synthetic-control approach  in two out of three outcomes analyzed, indicating that it is likely assigning the heterogeneity factors to the interactive fixed effects in the PO model, rather than in the treatment effect.

\lineskip=1ex \baselineskip 3.5ex
\bibliography{reference}

\baselineskip 4.5ex
\appendix
\setcounter{assumption}{0}
\renewcommand{\theassumption}{\thesection.\arabic{assumption}}
\setcounter{lemma}{0}
\renewcommand{\thelemma}{\thesection.\arabic{lemma}}
\setcounter{figure}{0}
\renewcommand{\thefigure}{A.\arabic{figure}}
\setcounter{table}{0}
\renewcommand{\thetable}{A.\arabic{table}}

\section{Proofs} \label{sec:Proofs}

We next establish all technical proofs of the paper in this section. We first consider the standard factor structure with time-varying factors and loadings that change across units.

To establish consistency for the IFE estimator with homogeneous treatment effects, we rely on the regularity conditions in Assumption~\ref{A:bai2009}, which are very similar to \pc{bai2009panel} Assumptions~A to~D. The main differences lie in our simplifying assumptions of independent and identically distributed errors and of deterministic factor structures. The latter are just for convenience \ca{bai2009panel}{see discussion in}. Let now $M$ denote a finite constant, $||\bs A||=\sqrt{\trace(\bs A'\bs A)}$, and $\mathcal{D}(\bs F)=\frac{1}{NT}\sum_{i=1}^N\bs{D}_i'\bs{M}_{\bs F}\bs{D}_i-\frac{1}{N^2T}\sum_{i=1}^N\sum_{j=1}^N\bs{D}_i'\bs{M}_{\bs F}\bs{D}_j\,a_{ij}$ with $a_{ij}=\bs\lambda_{0,i}'(\bs\Lambda_0'\bs\Lambda_0/N)^{-1}\bs\lambda_{0,j}$.\vskip 1em

\begin{assumption}[Regularity Conditions]\label{A:bai2009}
The following conditions hold:
    \begin{enumerate}[\textbf{(\alph*)}]
        \item $\inf_{\bs F\in\mathcal{F}}\mathcal{D}(\bs F)>0$, with $\mathcal{F}=\{\bs F:~\bs F'\bs F/T=\bs{I}_T\}$.
        \item\vspace*{-.2cm} $||\bs{F}_{0,t}||\le M$ and $\frac{1}{T}\sum_{t=1}^T\bs{F}_{0,t}\bs{F}_{0,t}'\to\bs{\Sigma_{F_0}}>\bs 0$ for some $k_0\times k_0$ matrix $\bs{\Sigma_{F_0}}$ as $T\to\infty$, whereas $||\bs\lambda_{0,i}||\le M$ and $\frac{1}{N}\sum_{i=1}^N\bs\lambda_{i,0}\bs\lambda_{i,0}'\to\bs{\Sigma_{\Lambda_0}}>\bs 0$ for some $k_0\times k_0$ matrix $\bs{\Sigma_{\Lambda_0}}$ as $N\rightarrow\infty$.
        \item\vspace*{-.2cm} $e_{i,t}$ is independent and identically distributed, with $\E(e_{i,t}^4)<\infty$.
    \end{enumerate}
\end{assumption}\vskip 1em

We now consider a series of lemmata in order to establish the proof of Proposition \ref{Prop:bad.control}.  Recall that $\mbox{SSE}(\alpha,\bs F)=\sum_{i=1}^N(\bs{Y}_i-\bs\alpha\bs{D}_i)'\bs M_F (\bs{Y}_i-\bs\alpha\bs{D}_i)$ and that $\bs{M}_{\bs F}=\bs{I}_{T}-\bs P_{\bs F}=\bs{I}_{T}-\bs F(\bs F'\bs F)^{-1}\bs F'=\bs{I}_{T}-\bs F\bs F'/T$.

\begin{lemma}\label{Lm:SSE}
    Consider the PO model in \eqref{eq:potential.outcomes}, with $k_0$ known and fixed. It then follows uniformly that
    \begin{eqnarray*}
        \mbox{SSE}(\alpha,\bs F)&=&\sum_{i=1}^N\bs{D}_i'\bs A_i\bs{M}_{\bs F}\bs A_i\bs{D}_i+\sum_{i=1}^N\bs\lambda_{0,i}'\bs{F}_0'\bs{M}_{\bs F}\bs{F}_0\bs\lambda_{0,i}+2\sum_{i=1}^N\bs{D}_i' \bs{A}_i'\bs{M}_{\bs F}\bs{F}_0\bs\lambda_{0,i}\\
        &&+\,\sum_{i=1}^N\bs{e}_i'\bs{M}_{\bs F}\bs{e}_i+o_p(NT)
    \end{eqnarray*}
    where $\bs A_i=\diag(\bs\alpha_i)-\alpha\bs{I}_T$.
\end{lemma}

\begin{proof}[Proof of Lemma~\ref{Lm:SSE}]
    We rewrite the PO model in \eqref{eq:potential.outcomes} as $\bs{Y}_i=\diag(\bs\alpha_i)\bs{D}_i+\bs{F}_0\bs\lambda_{0,i}+\bs{e}_i$ and then subtract $\alpha\bs{D}_i$ from both sides to obtain $\bs{Y}_i-\alpha\bs{D}_i=\bs A_i\bs{D}_i+\bs{F}_0\bs\lambda_{0,i}+\bs{e}_i$. It then follows that 
    \begin{eqnarray*}
    \mbox{SSE}(\alpha,\bs F) %
    &=&\sum_{i=1}^N\left(\bs{Y}_i-\alpha\bs{D}_i\right)'\bs{M}_{\bs F}\left(\bs{Y}_i-\alpha\bs{D}_i\right)\\
    &=&\sum_{i=1}^N(\bs{A}_i\bs{D}_i+\bs{F}_0\bs\lambda_{0,i}+\bs{e}_i)'\bs{M}_{\bs F}(\bs{A}_i\bs{D}_i+\bs{F}_0\bs\lambda_{0,i}+\bs{e}_i)\\
    &=&\sum_{i=1}^N\bs{D}_i'\bs A_i'\bs M_{\bs F}\bs A_i\bs{D}_i+\sum_{i=1}^N\bs\lambda_{0,i}'\bs{F}_0'\bs M_{\bs F}\bs{F}_0\bs\lambda_{0,i}+\sum_{i=1}^N\bs{e}_i'\bs M_{\bs F}\bs{e}_i\\
    &&+\,2\sum_{i=1}^N\bs{D}_i'\bs{A}_i'\bs M_{\bs F}\bs{F}_0\bs\lambda_{0,i}+2\sum_{i=1}^N\bs{D}_i'\bs{A}_i'\bs{M}_{\bs F}\bs{e}_i+2\sum_{i=1}^N\bs\lambda_{0,i}'\bs{F}_0'\bs M_{\bs F}\bs{e}_i.
    \end{eqnarray*}
    Let $\bs X_i=\bs A_i\bs D_i$, with elements $x_{i,t}=(\alpha_{i,t}-\alpha)\,d_{i,t}$ for $i=1,\ldots,N$ and $t=1,\ldots,T$. By the triangle inequality, $|\alpha_{i,t}-\alpha|\le|\alpha_{i,t}|+|\alpha|\le M+|\alpha|$ by Assumption~\ref{A:heterogeneity.tv.factor}, which ensures $|x_{i,t}|=|\alpha_{i,t}-\alpha|\,|d_{i,t}|\le M+|\alpha|$ given that $0\le d_{i,t}\le 1$. Lemma~A.1 in \cn{bai2009panel} thus applies to $\bs X_i$, and hence $\sup_{\bs F\in\mathcal{F}}\left\|\frac{1}{NT}\sum_{i=1}^N\bs X_i'\bs{M}_{\bs{F}}\bs{e}_i\right\|$ and $\sup_{\bs F\in\mathcal{F}}\left\|\frac{1}{NT} \sum_{i=1}^N\bs\lambda_{0,i}'\bs F_0'\bs{M}_{\bs F}\bs{e}_i\right\|$ are both of order $o_p(1)$, giving way to the desired expression for $\mbox{SSE}(\alpha,\bs F)$.
\end{proof}

We now define $S_{NT}(\alpha,\bs F)$ as the centered objective function:
\begin{displaymath}
S_{NT}(\alpha,\bs F)=\frac{1}{NT}\,\mbox{SSE}(\alpha,\bs F)-\frac{1}{NT}\sum_{i=1}^N\bs e_i'\bs M_{\bb F}\bs e_i.
\end{displaymath}
This does not alter the underlying minimization problem, as the second term in $S_{NT}(\alpha, \bs F)$ has no impact given it is independent of $\alpha$ and $\bs F$. Lemmata~\ref{Lm:tildeS} to~\ref{Lm:argmin_tildeS} establish the main arguments for the proof of Proposition~\ref{Prop:bad.control}. The first result shows that $S_{NT}(\alpha,\bs F)$ converges uniformly to $\tilde S_{NT}(\alpha,\bs F)$ for some bounded set of $(\alpha,\bs F)$, whereas the last two demonstrate that the minimum of $\tilde S_{NT}(\alpha,\bs F)$ is at $(\gamma,\bb F)$.  

\begin{lemma}\label{Lm:tildeS}
    The difference between $S_{NT}(\alpha,\bs F)$ and
    \begin{displaymath}
    \tilde S_{NT}(\alpha,\bs F)=\frac{1}{NT}\sum_{i=1}^N\Big(\bs{D}_i'\bs A_i\bs{M}_{\bs F}\bs A_i\bs{D}_i+\bs\lambda_{0,i}' \bs{F}_0'\bs{M}_{\bs F}\bs{F}_0 \bs\lambda_{0,i}+2\bs{D}_i'\bs A_i\bs{M}_{\bs F}\bs{F}_{0}\bs\lambda_{0,i}\Big)
    \end{displaymath}
    is uniformly small in probability: i.e., $S_{NT}(\alpha,\bs F)-\tilde S_{NT}(\alpha,\bs F)=o_p(1)$.
\end{lemma}

\begin{proof}[Proof of Lemma~\ref{Lm:tildeS}]
    It readily follows from Lemma~\ref{Lm:SSE} that
    \begin{eqnarray*}
        S_{NT}(\alpha,\bs F)&=&\frac{1}{NT}\,\sum_{i=1}^N\bs{D}_i'\bs A_i\bs{M}_{\bs F}\bs A_i\bs{D}_i+\frac{1}{NT}\sum_{i=1}^N\bs\lambda_{0,i}'\bs{F}_0'\bs{M}_{\bs F}\bs{F}_0\bs\lambda_{0,i}\\
        &&+\,\frac{2}{NT}\sum_{i=1}^N\bs{D}_i'\bs A_i\bs{M}_{\bs F}\bs{F}_0\bs\lambda_{0,i}+\frac{1}{NT}\sum_{i=1}^N\bs{e}_i'(\bs{P}_{\bs F}-\bs P_{\bb F}) \bs{e}_i+o_p(1)\\
        &=&\tilde S_{NT}(\alpha,\bs F)+\frac{1}{NT}\sum_{i=1}^N\bs{e}_i'(\bs{P}_{\bs F}-\bs P_{\bb F})\bs{e}_i+o_p(1)
    \end{eqnarray*}
    uniformly by Lemma~\ref{Lm:SSE}. It now suffices to see that $\sup_{\bs F\in\mathcal{F}}\left\|\frac{1}{NT}\sum_{i=1}^N\bs{e}_i'\bs{P}_{\bs F}\bs{e}_i\right\|=o_p(1)$ by \pc{bai2009panel} Lemma~A.1 and that $\bb F\in\mathcal{F}$.
\end{proof}

\begin{lemma}\label{Lm:tildeS_quad_form}
    Under Assumptions~\ref{A:sampling} to~\ref{A:heterogeneity.tv.factor} for the PO model in~\eqref{eq:potential.outcomes}, with $k_0$ fixed and known, $\tilde S_{NT}(\beta,\bs F)=\beta^2\bb D(\bs F)+\bs\theta'\bs{B}\bs\theta$, with $\beta=\gamma-\alpha$, $\bs\theta=\mathrm{vec}(\bs M_F\bb F)+\bs{B}^{-1}\bs{C}\beta$, $\bs{B}=(\bb\Lambda'\bb\Lambda/N)\otimes\bs{I}_T$, and $\bs{C}=(1/NT)\sum_{i=1}^N\bb\lambda_i\otimes\bs M_F\bs D_i$.
\end{lemma}

\begin{proof}[Proof of Lemma~\ref{Lm:tildeS_quad_form}]
    By Assumption~\ref{A:heterogeneity.tv.factor}, $\diag(\bs\alpha_i)\bs{D}_i=\gamma\bs{D}_i+\diag(\bs F_\alpha\bs\lambda_{\alpha,i})\bs{D}_i$ and hence $\bs A_i\bs D_i=(\gamma-\alpha)\bs{D}_i+\diag(\bs F_\alpha\bs\lambda_{\alpha,i})\bs{D}_i=(\gamma-\alpha)\bs{D}_i+\tilde{\bs{F}}_\alpha \tilde{\bs{\lambda}}_{\alpha,i}$, with $\tilde{\bs F}_\alpha=(\bs 0_{T_0\times k_\alpha}',\bs{F}_{\alpha,t\ge T_0}')'$ and $\tilde{\bs{\lambda}}_{\alpha,i}=\bs\lambda_{\alpha,i}$ if $i\in\mathcal{T}$, $\bs 0_{k_\alpha\times 1}$ otherwise. This holds because\\
    \begin{displaymath}
    \begin{split}
        \diag(\bs F_\alpha\bs\lambda_{\alpha,i})\bs{D}_i &= \begin{bmatrix}
            d_{i,1}\bs\lambda_{\alpha,i}'\bs{F}_{\alpha,1}\\
             \vdots\\
             d_{i,T}\bs\lambda_{\alpha,i}'\bs{F}_{\alpha,T}
        \end{bmatrix} =
        \begin{bmatrix}
            (0\times\bs\lambda_{\alpha,i})'\bs{F}_{\alpha,1}\\
            \vdots\\
            (0\times\bs\lambda_{\alpha,i})'\bs{F}_{\alpha,T_0}\\
            (d_i\bs\lambda_{\alpha,i})'\bs{F}_{\alpha,T_0+1}\\
            \vdots\\
            (d_i\bs\lambda_{\alpha,i})'\bs{F}_{\alpha,T} 
        \end{bmatrix} = \tilde{\bs F}_\alpha\tilde{\bs{\lambda}}_{\alpha,i}.\\
    \end{split}
    \end{displaymath}
    As a result,
    \begin{eqnarray*}
        \tilde S_{NT}(\alpha,\bs F)&=&\frac{1}{NT}\sum_{i=1}^N\big((\gamma-\alpha)\bs{D}_i+\tilde{\bs{F}}_\alpha\tilde{\bs{\lambda}}_{\alpha,i} \big)'\bs{M}_{\bs F}\big((\gamma-\alpha)\bs{D}_i+\tilde{\bs{F}}_\alpha\tilde{\bs{\lambda}}_{\alpha,i}\big)\\
        &&+\,\frac{1}{NT}\sum_{i=1}^N\bs\lambda_{0,i}'\bs{F}_0'\bs{M}_{\bs F}\bs{F}_0\bs\lambda_{0,i}+\frac{2}{NT}\sum_{i=1}^N\bs{D}_i' \big((\gamma-\alpha)+\tilde{\bs{F}}_\alpha\tilde{\bs{\lambda}}_{\alpha,i}\big)'\bs{M}_{\bs F}\bs{F}_0\bs\lambda_{0,i}\\
        &=&(\gamma-\alpha)^2\,\frac{1}{NT}\sum_{i=1}^N\bs{D}_i'\bs{M}_{\bs F}\bs{D}_i+2(\gamma-\alpha)\,\frac{1}{NT}\sum_{i=1}^N\bs{D}_i' \bs{M}_{\bs F}\tilde{\bs{F}}_\alpha\tilde{\bs{\lambda}}_{\alpha,i}\\
        &&+\,\frac{1}{NT}\sum_{i=1}^N\tilde{\bs{\lambda}}_{\alpha,i}'\tilde{\bs{F}}_\alpha'\bs{M}_{\bs F}\tilde{\bs{F}}_\alpha \tilde{\bs{\lambda}}_{\alpha,i}+\frac{1}{NT}\sum_{i=1}^N\bs\lambda_{0,i}'\bs{F}_0'\bs{M}_{\bs F}\bs{F}_0\bs\lambda_{0,i}\\
        &&+\,\frac{2}{NT}\sum_{i=1}^N \tilde{\bs{\lambda}}_{\alpha,i}'\bs{\tilde{F}}_\alpha'\bs{M}_{\bs F}\bs{F}_0\bs\lambda_{0,i}+2(\gamma-\alpha)\,\frac{1}{NT}\sum_{i=1}^N\bs{D}_i'\bs{M}_{\bs F}\bs{F}_0\bs\lambda_{0,i}.
    \end{eqnarray*}
    However, it follows from $(\bs{F}_0\bs\lambda_{0,i}+\tilde{\bs{F}}_\alpha\tilde{\bs{\lambda}}_{\alpha,i})'\bs{M}_{\bs F}(\bs{F}_0\bs\lambda_{0,i}+\tilde{\bs{F}}_\alpha\tilde{\bs{\lambda}}_{\alpha,i})=\tilde{\bs{\lambda}}_{\alpha,i}'\tilde{\bs{F}}_\alpha'\bs{M}_{\bs F} \tilde{\bs{F}}_\alpha\tilde{\bs{\lambda}}_{\alpha,i}+\bs\lambda_{0,i}'\bs{F}_0'\bs{M}_{\bs F}\bs{F}_0 \bs\lambda_{0,i}+2\,\tilde{\bs{\lambda}}_{\alpha,i}'\tilde{\bs{F}}_\alpha'\bs{M}_{\bs F}\bs{F}_0\bs\lambda_{0,i}$ that
    \begin{eqnarray*}
        \tilde S_{NT}(\alpha,\bs F)&=&(\gamma-\alpha)^2\,\frac{1}{NT}\sum_{i=1}^N\bs{D}_i'\bs{M}_{\bs F}\bs{D}_i+\sum_{i=1}^N(\bs{F}_0 \bs\lambda_{0,i}+\tilde{\bs{F}}_\alpha\tilde{\bs{\lambda}}_{\alpha,i})'\bs{M}_{\bs F}(\bs{F}_0\bs\lambda_{0,i}+\tilde{\bs{F}}_\alpha \tilde{\bs{\lambda}}_{\alpha,i})\\
        &&+\,2(\gamma-\alpha)\,\frac{1}{NT}\sum_{i=1}^N\bs{D}_i^\prime\bs{M}_{\bs F}(\tilde{\bs{F}}_\alpha\tilde{\bs{\lambda}}_{\alpha,i}+\bs{F}_0\bs\lambda_{0,i})\\
        &=&\frac{(\gamma-\alpha)^2}{NT}\sum_{i=1}^N\bs{D}_i'\bs{M}_{\bs F}\bs{D}_i+\frac{1}{NT}\sum_{i=1}^N\bb\lambda_i'\bb F'\bs{M}_{\bs F}\bb F\bb\lambda_i+\frac{2(\gamma-\alpha)}{NT}\sum_{i=1}^N\bs{D}_i^\prime\bs{M}_{\bs F}\bb F\bb\lambda_i\\
        &=&\frac{\beta^2}{NT}\sum_{i=1}^N\bs D_i^\prime\bs{M}_{\bs F}\bs D_i+\trace\left[\left(\frac{\bb{F}'\bs{M}_{\bs F}\bb{F}}{T}\right) \left(\frac{\bb\Lambda'\bb\Lambda}{N}\right)\right]+\frac{2\beta}{NT}\sum_{i=1}^N\bs D_i^\prime\bs{M}_{\bs F}\bb{F}\bb\lambda_i.
    \end{eqnarray*}
    The same argument as in \pc{bai2009panel} Proposition~1 then yields
    \begin{equation} \label{eq:tildeS_quadratic_form}
        \tilde S_{NT}(\beta,\bs F)= \beta^2 \bb D(\bs F) + \bs \theta^\prime \bs B \bs \theta
    \end{equation}
    with $\bb D(\bs F)$ as in Assumption~\ref{A:heterogeneity.tv.factor}(iii) and $\bs\theta$ as in the enunciate. In fact, the only difference with respect to \cn{bai2009panel} is that we consider a single covariate $\bs D_i$, so that $\beta$ is a scalar rather than a vector.
\end{proof}

\begin{lemma}\label{Lm:argmin_tildeS}
Under Assumptions~\ref{A:sampling} to~\ref{A:heterogeneity.tv.factor} for the PO model in~\eqref{eq:potential.outcomes}, with $k_0$ fixed and known, $(\gamma,\bb F\bs H)=\argmin_{\alpha,\bs F}\tilde S_{NT}(\alpha,\bs F)$ for some rotation matrix $\bs H$. 
\end{lemma}

\begin{proof}[Proof of Lemma~\ref{Lm:argmin_tildeS}]
    We know from Lemma~\ref{Lm:tildeS_quad_form} that $\tilde S_{NT}(\alpha,\bs F)=\beta^2\bb D(\bs F)+\bs\theta'\bs B\bs\theta$. In view that $\bs{B}=(\bb\Lambda'\bb\Lambda/N)\otimes\bs{I}_T$ is positive definite and $\bb D(\bs F)\ge 0$ for all $\bs F\in\mathcal{F}$, $\tilde S_{NT}(\alpha,\bs F)\ge 0$ for all $(\alpha,\bs F)$, with equality holding at $(\gamma,\bb F)$. If we let $(\beta_*,\bs F_*)=\argmin\left\{\beta^2\bb D(\bs F)+\bs\theta'\bs B\bs\theta\right\}$, it then remains to show that $\tilde S_{NT}(\gamma-\beta_*,\bs F_*)=0$ if and only if $(\beta_*,\bs F_*)=(0,\bs H\bb F)$. To do so, we introduce $\bs\theta_*=\mathrm{vec}(\bs M_{F_*}\bb F)+\bs B^{-1}\bs C_*\beta_*$, with $\bs C_*=(1/NT)\sum_{i=1}^N\bb\lambda_i\otimes\bs M_{F_*}\bs D_i$. Recall that $(\beta_*,\bs F_*)=\argmin\tilde S_{NT}(\gamma-\beta,\bs F)$, and hence $\tilde S_{NT}(\gamma-\beta_*,\bs F_*)=0$ provided that both $\bs\theta_*$ and $\beta_*^2\bb D(\bs F_*)$ must equal zero. There are two cases to consider: ($a$)~$\bs\theta_*=\bs 0$ and $\beta_*=0$; and ($b$)~$\bs\theta_*=\bs 0$ and $\bb D(\bs F_*)=0$.
    \begin{enumerate}[($a$)]
        \item If $\beta_*=0$, then $\bs\theta_*=\mathrm{vec}(\bs M_{F_*}\bb F)$ must equal zero to ensure that $\tilde S_{NT}(\beta_*,\bs F_*)=0$, which would imply imply that $\bb F$ belongs to the linear subspace spanned by $\bs F_*$. As such, there exists a rotation matrix $\bs H$ such that $\bs F_*=\bb F\bs H$ given that $\bb D(\bs F_*)>0$.
        \item Assume, by contradiction, that $\bb D(\bs F_*)=0$ and $\bb D(\bb F)>0$. Let $\bs d_{1:T}=(d_1,d_2,\ldots,d_T)'$ and consider the scalar $a_{ij}=\bb\lambda_i'\left(\bb\Lambda'\bb\Lambda/N\right)^{-1}\bb\lambda_j$. It follows from $d_{i,t}=d_id_t$, with $d_i^n=d_i$ for any $n\ge 0$, that
        \begin{eqnarray*}
            \bb D(\bs F_*)&=&\frac{1}{NT}\sum_{i=1}^N\bs{D}_i'\bs{M}_{\bs F^*}\bs{D}_i-\frac{1}{N^2T}\sum_{i=1}^N\sum_{j=1}^N\bs{D}_i'\bs{M}_{\bs F^*} \bs{D}_j a_{ij}\\
            &=&\frac{1}{NT}\sum_{i=1}^N(d_i\bs d_{1:T})'\bs{M}_{\bs F^*}(d_i\bs d_{1:T})-\frac{1}{N^2T}\sum_{i=1}^N\sum_{j=1}^N a_{ij}(d_i\bs d_{1:T})' \bs{M}_{\bs F^*}(d_j\bs d_{1:T})\\
            &=&A(\bs F_*)\breve{B},
        \end{eqnarray*}
        where $A(\bs F_*)=\frac{1}{T}\,\bs d_{1:T}'\bs{M}_{\bs F^*}\bs d_{1:T}$ and $\breve{B}=\Bigg[\frac{1}{N}\sum\limits_{i=1}^N d_i-\frac{1}{N^2} \sum\limits_{1\le i,j\le N}d_id_j\bb\lambda_i'(\bb\Lambda'\bb\Lambda/N)^{-1}\bb\lambda_j\Bigg]$. Given the quadratic nature of $A(\bs F_*)$, $\bb D(\bs F_*)=0$ if and only if $A(\bs F_*)=0$ in view that $\breve{B}=0$ only if $\bb D(\bs F)=0$ for all $\bs F$, which would contradict Assumption~\ref{A:heterogeneity.tv.factor}(iii). In turn, $\bs d_{1:T}'\bs{M}_{\bs F_*}\bs d_{1:T}=0$ if $A(\bs F_*)=0$, implying that $\bs M_{\bs F_*}\bs D_i=\bs 0$ for all $i$. This means that $\bs C_*=\bs 0$ and hence $\bs M_{\bs F_*}\bb F=\bs 0$ to ensure that $\bs\theta_*=\bs 0$. This implies that $\bs F_*=\bb F\bs H$, thereby contradicting the condition $\bb D(\bs F_*)=\bb D(\bb F)>0$ in Assumption~\ref{A:heterogeneity.tv.factor}(iii).
    \end{enumerate}
    This completes the proof.
\end{proof}

\begin{proof}[Proof of Proposition 1]
     Lemma~\ref{Lm:tildeS} dictates that $S_{NT}(\alpha,\bs F)-\tilde S_{NT}(\alpha,\bs F)=o_p(1)$, whereas Lemma \ref{Lm:argmin_tildeS} ensures that $\tilde S_{NT}(\alpha,\bs F)>0$ holds for either $\alpha\ne\gamma$ or $\bs F\ne\bb F\bs H$ given that $(\gamma,\bb F\bs H)=\argmin_{\alpha,\bs F}\tilde S_{NT}(\alpha,\bs F)$. Altogether, this means that the value of $\hat\alpha$ that minimizes $S_{NT}(\alpha,\bs F)$ converges in probability to $\gamma$. In contrast, consistency does not necessarily ensue for $\hat{\bs{F}}$ because the number of elements in $\bb F$ grows as $T\to\infty$. It is nonetheless possible to show that $||\bs P_{\hat{\bs{F}}}-\bs P_{\bb F}||\xrightarrow{p}0$ by following the same steps as in the proof of Proposition 1, part (ii), in \cn{bai2009panel}, given that Assumption~\ref{A:heterogeneity.tv.factor} ensures $\bb F$ and $\bb\Lambda$ satisfy every necessary condition. 
\end{proof}

The proofs of Proposition~\ref{Prop:bad.control} and Lemmata~\ref{Lm:SSE} to~\ref{Lm:argmin_tildeS} assume no observable controls. We impose this restriction solely for convenience, but it is quite straightforward to extend the asymptotic theory to consider a $p$-dimensional vector of covariates  $\bs Z_{i,t}$ such that $\E||\bs Z_{i,t}||^4<M$ for some finite constant $M$, with partial effects $\bs\delta_0$. Let $\bs X_i=(\bs D_i,\bs Z_i)$ and $\bs\beta=(\alpha-\gamma,\bs\delta-\bs\delta_0)'$, and then rewrite the centered objective function as $S_{NT}(\bs\beta,\bs F)=\tilde{S}_{NT}(\bs\beta,\bs F)+o_p(1)$ uniformly on $(\bs\beta,\bs F)$. Identification now holds for $\bs\beta=\bs 0$ or, equivalently, $\alpha=\gamma$ and $\bs\delta=\bs\delta_0$, with all proofs ensuing analogously by applying the Frisch–Waugh–Lovell theorem.

\section{Monte Carlo setting} \label{App:setting.sim}

This section establishes the setting for the data generation processes in the Monte Carlo study we carry out in Section~\ref{sec:MC.bad.control}. We standardize the heterogeneous treatment effects to impose that $\bar\alpha=2$. More specifically, we set $\bs\lambda_{\alpha,i}=\bs\delta_{\alpha,i}-\bar{\bs\delta}_\alpha+(-1,1)'$ for treated units ($i\in\mathcal{T}$) and $\bs{F}_{\alpha,t}=\bs{f}_{\alpha,t}-\bar{\bs f}_\alpha+(1,2)'$ for all post-treatment periods ($t>T_0$), where $\bar{\bs\delta}_\alpha=\frac{1}{N_1}\sum_{i\in\mathcal{T}}\bs\delta_{\alpha,i}$, and $\bar{\bs F}_\alpha=\frac{1}{T_1}\sum_{t>T_0}\bs{F}_{\alpha,t}$. We draw $\bs{F}_{\alpha,t}$ and $\bs\delta_{\alpha,i}$ from independent standard normal distributions, holding their values fixed across $R=10{,}000$ replications. It then follows that $\bar\alpha=\gamma+\frac{1}{N_1T_1}\sum_{i\in\mathcal{T}}\sum_{t=T_0+1}^T\bs\lambda_{\alpha,i}'\bs{F}_{\alpha,t}=1-1+2=2$ in every replication. The heterogeneity factors and their loadings are independent of the errors, as well as of the PO factors and loadings in \eqref{eq:potential.outcomes}.

We also keep the factor structure in the PO model fixed across replications. We assume a single common factor that follows an AR(1) process: $F_{0,t}=-2+\rho_t\,F_{0,t-1}+e_{F,t}$, where $\rho_t=0.25+0.50\,\bs{1}(t\ge T_0+1)$ and $e_{F,t}\sim N(0,1-\rho_t^2)$. This ensures that the PO factor has different means in the pre-treatment and post-treatment periods, while keeping a constant unit variance. As in \cn{xu2017generalized}, we draw the factor loading from a uniform $U(-\sqrt{3},\sqrt{3})$ for untreated units, and from a uniform $U(-\sqrt{3/4},\sqrt{3})$ for treated units. This ensures that treated and untreated units have different averages while satisfying the consistency conditions for every estimator we consider, such as the common support assumption in the synthetic-control approach.

Finally, we carry out the simulations considering a balanced panel in which there are not only as many treated units as untreated units, but also as many pre-treatment periods as post-treatment periods. In each replication, we allocate the treatment to the first $N_1=N/2$ units as from $T_0+1=T/2+1$ (i.e., the post-treatment period comprises the last $T_1=T/2$ observations in the time series of the treated units). We focus exclusively on time-varying factors in Table~\ref{tab:simulation_bad_control} to highlight the problem with bad controls. We relegate to Appendix~\ref{App:sim.time-invariant} the simulations with time-invariant factors, which evince multimodality in the sampling distribution of the IFE estimator.

\section{Simulations with invariant factor structures} \label{App:sim.time-invariant}

We consider an additional complication that may arise when heterogeneous treatment effects include invariant factors (which are assumed away in Assumption \ref{A:heterogeneity.tv.factor}). In this case, the interaction between the invariant factors and the treatment dummy becomes perfectly collinear with the treatment dummy. To assess the consequences of this issue, we consider a set of Monte Carlo simulations with DGPs featuring heterogeneous treatment effects that are invariant over time.  

Potential outcomes follow the model in~\ref{eq:potential.outcomes}, with treatment effects constant over time, i.e., $\alpha_{i,t}=\alpha_i$ for all $i$ and $t$. In particular, $\alpha_i=\sum_{j=1}^{k_\alpha}\mu_j b_i^{(j)}$, where the $\mu_j$ are scalar constants such that $\mu_1=1/2$, $\mu_2=1$, and $\mu_3=-1/2$, and the $b_i^{(j)}$ are independent Bernoulli random variables with probability $1/2$. This implies ATT values of $1/4$ for $k_\alpha=1$, $3/4$ for $k_\alpha=2$, and $1/2$ for $k_\alpha=3$.  Such heterogeneity in treatment effects admits a linear factor representation with $\bs{F}_{\alpha,t}=(1,\ldots,1)'$ for all $t$ and $\bs{\lambda}_{\alpha,i}=(\mu_1 b_i^{(1)},\ldots,\mu_{k_\alpha} b_i^{(k_\alpha)})'$. For the potential outcomes when untreated, the data-generating process includes a single common factor ($k_0=1$) and idiosyncratic errors drawn from mutually independent standard normal distributions.

Figure~\ref{fig:time-invariant.heterogeneity} illustrates the finite-sample behavior of the estimators under these DPGs in settings with $k_\alpha\in\{1,2,3\}$. In these DGPs, the IFE estimator of $\bar\alpha$ exhibits multimodal distributions. The problem is more severe for realizations of the data in which $\bs{D}(\hat{\bs{F}})\approx 0$, meaning that the estimated factor structure is nearly collinear with the treatment dummy. When we condition on realizations with low values of $\bs{D}(\hat{\bs{F}})$, we observe greater dispersion in the IFE estimator and, in some cases, a bimodal distribution. Conversely, when $\bs{D}(\hat{\bs{F}})$ is larger, the IFE estimator tends to be closer to the true ATT. Nevertheless, when we consider the overall distribution of $\hat\alpha$, it is clear that the IFE estimator is biased in the presence of this type of treatment effect heterogeneity.

\begin{figure}[!htbp]
    \centering\includegraphics[width=\linewidth]{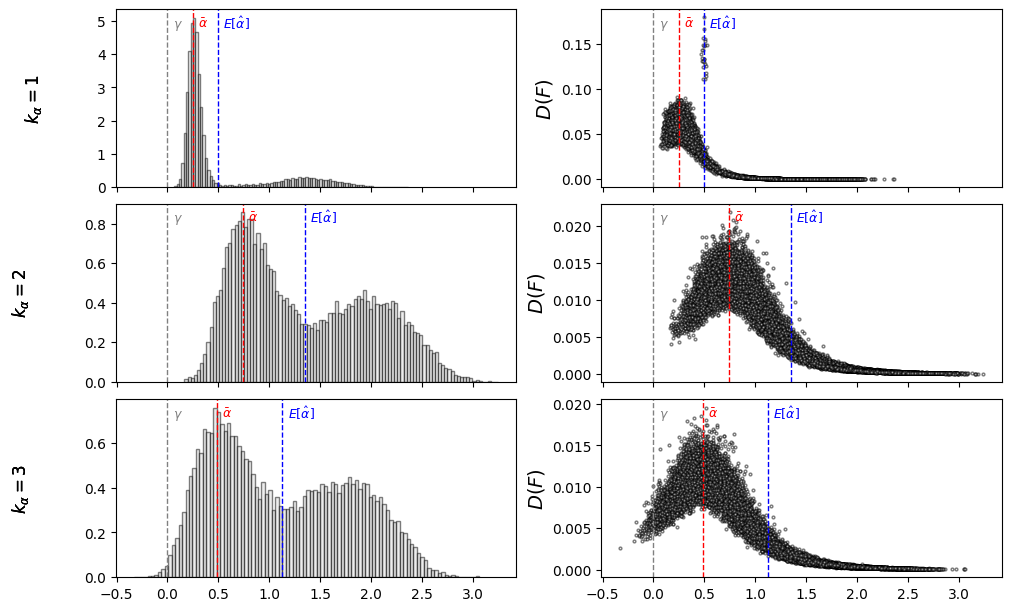}
    \caption{Simulations with time-invariant factors}
    \label{fig:time-invariant.heterogeneity}
\end{figure}

\section{Simulations with dynamic specification} \label{App:sim.dynamic}

To allow treatment effects to vary over time, one could well introduce interactions between treatment and time indicators for each post-treatment period. If the heterogeneity in treatment effects were of the form $\alpha_{i,t}=\psi_t$, adopting a dynamic specification would effectively make the treatment effects homogeneous:
\begin{equation} \label{eq:Dynamic_specification}
y_{i,t}=\sum_{s>T_0}\psi_s\,\mathbf{1}(t=s)\,d_{i,t}+\bs{\lambda}_{0,i}'\bs{F}_{0,t}+u_{i,t},
\end{equation}
where $u_{i,t}$ is a white noise. In this case, the IFE estimator is consistent for $(\psi_{T_0+1},\ldots,\psi_T)$. As before, augmenting heterogeneity with a purely idiosyncratic component $a_{i,t}$ does not affect this conclusion, given that it does not admit a linear factor representation. Under dynamic specifications, bad controls arise only if the heterogeneity in treatment effects possesses a linear factor structure on top of $\psi_t$.

To illustrate this case, we run Monte Carlo simulations using the same design as in \cn{xu2017generalized}. Treatment effects satisfy $\alpha_{i,t}=\psi_t+a_{i,t}$, where $\psi_t=t-T_0$ and $a_{i,t}\sim N(0,5)$ is totally idiosyncratic without any factor representation. The outcome is given by $$y_{i,t}=\alpha_{i,t}\,d_{i,t}+5+\bs\lambda_{0,i}'\bs F_{0,t}+\zeta_t+\mu_i+x_{i,t}^{(1)}+3\,x_{i,t}^{(2)} +e_{i,t},$$
where $x_{i,t}^{(j)}=1+\bs\lambda_{0,i}'\bs F_{0,t}+\bs\lambda_{0,i}'(1,1)+\bs F_{0,t}'(1,1)+\eta_{i,t}^{(j)}$, with $\eta_{i,t}^{(j)}$ denoting independent standard normal noises for $j\in\{1,2\}$. The common factors $\bs{F}_{0,t}$ and time fixed effects $\zeta_t$ are also independent standard normal, whereas factor loadings and unit fixed effects are uniform $U(-\sqrt{3},\sqrt{3})$. Finally, we contemplate different sample sizes $(N_0,N_1,T_0,T_1)$, resting on $R=10{,}000$ replications.

\input{Tables/Dynamic}

Under the dynamic specification described in Equation~\eqref{eq:Dynamic_specification}, the IFE estimator is consistent for the ATT parameter at any post-treatment period in that $\hat\psi_t\xrightarrow{p}\psi_t$ for all $t>T_0$. Accordingly, the IFE estimator of the average post-treatment effect $\hat\alpha=(1/T_1)\sum_{t>T_0}\hat\psi_t$ is also consistent, converging in probability to $\bar\alpha$. Table~\ref{tab:dynamic} reports the behavior of the IFE estimators under the dynamic specification for both the average treatment effect (i.e., $\bar\alpha$) and the ATT at time $t=T_0+10$ (i.e., $\psi_{T_0+10}=10)$. It is apparent that the bias and standard error of both IFE estimators shrink to zero as the sample size grows.

The behavior of the IFE estimator is sharply different under the static specification in~\ref{eq:SSE_min}. Bias arises if the proportions of treated units and treated periods are non-negligible, in line with the bad-control problem we describe in Section \ref{sec:heterogenous}. It is also interesting that the standard error in this last scenario is larger than in scenarios with fewer post-treatment periods and fewer treated units. This happens because, in those cases, treatment-effect heterogeneity does not generate enough variation for the IFE estimator to capture it as part of the factor model (see Remark~\ref{rmk:non-negligible}). In contrast, when the proportions of treated units and treated periods are non-negligible, the factor estimates absorb the treatment effect heterogeneity, thereby increasing the multicollinearity with the treatment dummy. This implies larger standard errors, as we discuss in Remark~\ref{rmk:fixed.effect}.
\end{document}

%% file: Tables/Comparison.tex
\begin{landscape}
\begin{table}
\caption{ATT estimation under homogeneous and heterogeneous treatment effects}
\label{tab:simulation_bad_control}
\centering
\scalebox{0.7}{
\begin{tabular}{llcccccccccc}
\multicolumn{12}{p{29.5cm}}{{\footnotesize We simulate the PO model in \eqref{eq:potential.outcomes}, with $F_{0,t}=-2+\rho_t\,F_{0,t-1}+e_{F,t}$, where $\rho_t=0.25+0.50\,\bs{1}(t\ge T_0+1)$ and $e_{F,t}\sim N(0,1-\rho_t^2)$, and $\bs{\lambda}_{0,i}$ coming from $U(-\sqrt{3},\sqrt{3})$ for untreated units and from $U(-\sqrt{3/4},\sqrt{3})$ for treated units. We set $\alpha_{i,t}=\alpha=2$ for all units and time periods in the case of homogeneous treatment effects, whereas $\alpha_{i,t}=1+\bs\lambda_{\alpha,i}'\bs{F}_{\alpha,t}$ if treatment effects are heterogeneous. In particular, $\bs\lambda_{\alpha,i}=\bs\delta_{\alpha,i}-\bar{\bs\delta}_\alpha+(-1,1)'$ for $i\in\mathcal{T}$, whereas $\bs{F}_{\alpha,t}=\bs{f}_{\alpha,t}-\bar{\bs f}_\alpha+(1,2)'$ for $t>T_0$, where $\bs{f}_{\alpha,t}\sim N(0,1)$, $\bs\delta_{\alpha,i}\sim N(0,1)$, $\bar{\bs\delta}_\alpha=\frac{1}{N_1}\sum_{i\in\mathcal{T}}\bs\delta_{\alpha,i}$, and $\bar{\bs f}_\alpha=\frac{1}{T_1}\sum_{t>T_0}\bs{f}_{\alpha,t}$. All standard normal variates are mutually independent, and we hold the values of he factors and their loadings fixed across replications. Finally, we allocate the treatment to the first $N_1=N/2$ units as from $T_0+1=T/2+1$, thereby forming a balanced panel. Each column displays the mean and standard deviation (in parentheses) of the ATT estimates over 10,000 simulations.}}\\
\toprule
&     & \multicolumn{5}{c}{homogeneous treatment effects} & \multicolumn{5}{c}{heterogeneous  treatment effects} \\
\cmidrule(lr){3-7} \cmidrule(lr){8-12}
$T$ & $N$ & IFE & GSC & SC & DSC & SDiD & IFE & GSC & SC & DSC & SDiD\\
\midrule
50      & 50    & 2.001 (0.084) & 2.003 (0.093) & 2.030 (0.088) & 2.018 (0.090) & 2.005 (0.111) & 0.988 (0.264) & 2.018 (0.111) & 2.148 (0.101) & 2.104 (0.103) & 2.029 (0.126) \\
        & 100   & 2.000 (0.059) & 2.003 (0.061) & 2.037 (0.058) & 2.010 (0.062) & 2.004 (0.089) & 0.997 (0.185) & 2.012 (0.069) & 2.072 (0.067) & 2.048 (0.069) & 2.015 (0.097) \\
        & 200   & 2.000 (0.043) & 2.009 (0.044) & 2.052 (0.043) & 2.025 (0.045) & 2.013 (0.070) & 1.000 (0.111) & 2.012 (0.046) & 2.072 (0.044) & 2.043 (0.046) & 2.009 (0.072) \\
        & 1000  & 2.000 (0.020) & 2.014 (0.021) & 2.102 (0.020) & 2.047 (0.021) & 2.002 (0.083) & 1.000 (0.055) & 2.011 (0.021) & 2.051 (0.020) & 2.029 (0.021) & 1.999 (0.083) \\
100     & 50    & 2.000 (0.066) & 2.010 (0.078) & 2.149 (0.070) & 2.085 (0.074) & 2.030 (0.084) & 0.996 (0.209) & 2.012 (0.083) & 2.102 (0.073) & 2.088 (0.075) & 2.040 (0.083) \\
        & 100   & 2.001 (0.058) & 2.032 (0.072) & 2.164 (0.062) & 2.171 (0.063) & 2.077 (0.080) & 0.999 (0.115) & 2.004 (0.047) & 2.067 (0.045) & 2.038 (0.047) & 2.009 (0.059) \\
        & 200   & 2.000 (0.031) & 2.006 (0.033) & 2.070 (0.032) & 2.040 (0.033) & 2.007 (0.047) & 0.999 (0.099) & 2.026 (0.046) & 2.102 (0.043) & 2.108 (0.043) & 2.039 (0.062) \\
        & 1000  & 2.000 (0.014) & 2.006 (0.015) & 2.054 (0.015) & 2.031 (0.015) & 2.005 (0.041) & 1.000 (0.039) & 2.009 (0.016) & 2.050 (0.015) & 2.033 (0.016) & 2.003 (0.043) \\
200     & 50    & 2.000 (0.044) & 2.004 (0.052) & 2.087 (0.046) & 2.056 (0.048) & 2.019 (0.051) & 1.000 (0.122) & 2.007 (0.060) & 2.124 (0.053) & 2.105 (0.054) & 2.034 (0.057) \\
        & 100   & 2.000 (0.035) & 2.009 (0.042) & 2.097 (0.038) & 2.094 (0.039) & 2.024 (0.045) & 1.000 (0.082) & 2.002 (0.034) & 2.056 (0.034) & 2.037 (0.035) & 2.007 (0.039) \\
        & 200   & 2.000 (0.023) & 2.005 (0.025) & 2.054 (0.025) & 2.045 (0.025) & 2.010 (0.031) & 0.999 (0.072) & 2.001 (0.022) & 2.043 (0.022) & 2.013 (0.023) & 2.002 (0.029) \\
        & 1000  & 2.000 (0.011) & 2.006 (0.012) & 2.045 (0.012) & 2.036 (0.012) & 2.008 (0.020) & 1.000 (0.025) & 2.007 (0.012) & 2.072 (0.012) & 2.048 (0.013) & 2.009 (0.021) \\
1000    & 50    & 2.000 (0.021) & 2.001 (0.024) & 2.085 (0.023) & 2.068 (0.023) & 2.027 (0.022) & 1.000 (0.048) & 2.000 (0.021) & 2.040 (0.019) & 1.988 (0.022) & 1.996 (0.020) \\
        & 100   & 2.000 (0.014) & 2.000 (0.015) & 2.043 (0.014) & 2.014 (0.016) & 2.003 (0.015) & 1.000 (0.036) & 2.001 (0.017) & 2.057 (0.017) & 2.051 (0.017) & 2.013 (0.017) \\
        & 200   & 2.000 (0.010) & 2.001 (0.011) & 2.035 (0.011) & 2.032 (0.011) & 2.006 (0.012) & 1.000 (0.025) & 2.001 (0.011) & 2.038 (0.011) & 2.025 (0.011) & 2.004 (0.011) \\
        & 1000  & 2.000 (0.005) & 2.001 (0.005) & 2.023 (0.005) & 2.019 (0.005) & 2.002 (0.006) & 1.000 (0.011) & 2.001 (0.005) & 2.027 (0.005) & 2.023 (0.005) & 2.002 (0.006) \\
\bottomrule
\end{tabular}}
\end{table}
\end{landscape}

%% file: Tables/Empirical.tex
\newcolumntype{P}[1]{>{\centering\arraybackslash}p{#1}}

\begin{table}[!ht]
    \centering
    \caption{Enterprise Zone Program effects on unemployment entry and exit rates}
    \begin{tabular}{p{3.5cm}rrcr}
    \multicolumn{5}{p{12cm}}{{\footnotesize Outputs are on a logarithmic scale. The p-values in brackets refer to the permutation-based placebo tests. IFE refers to the interactive fixed effects estimator \cite{bai2009panel}; SC denotes the standard synthetic-control estimator \cite{abadie2010synthetic}; DSC refers to the demeaned SC estimator \cite{ferman2021synthetic}; GSC denotes the generalized SC estimator \cite{xu2017generalized}; and SDiD corresponds to the synthetic difference-in-differences estimator \cite{arkhangelsky2021synthetic}.}}\\
    \toprule
           & \multicolumn{2}{c}{exit rate}\\
    \cline{2-3}
    estimators &       to a job & unknown reasons && entry rate\\
    \midrule\\ [-2.5em]
    IFE         & 0.036     & 0.003     && 0.004 \\ [-1.25em]
                & [0.097]   & [0.906]   && [0.824] \\
    GSC         & -0.172    & 0.053     && -0.016 \\ [-1.25em]
                & [0.001]   & [0.083]   && [0.360] \\
    SC          & -0.032    & 0.051     && 0.007 \\ [-1.25em]
                & [0.349]   & [0.015]   && [0.155] \\
    DSC         & -0.015    & 0.044     && 0.037 \\ [-1.25em]
                & [0.340]   & [0.066]   && [0.241] \\
    SDiD        & -0.002    & 0.048     && 0.039 \\ [-1.25em]
                & [0.947]   & [0.023]   && [0.015] \\
    \bottomrule
    \end{tabular}\label{tab:empirical}
\end{table}

%% file: Tables/Dynamic.tex
\begin{table}[!ht]
    \centering\caption{Dynamic Specification with $\alpha_{it} = \delta_t + e_{i,t}^\alpha$}
    \begin{tabular}{llllcccc}
    \multicolumn{8}{p{.875\textwidth}}{{\footnotesize The PO model reads $y_{i,t}=\alpha_{i,t}\,d_{i,t}+5+\bs\lambda_{0,i}'\bs F_{0,t}+\zeta_t+\mu_i +x_{i,t}^{(1)}+3x_{i,t}^{(2)}+e_{i,t}$, where $\alpha_{i,t}=\psi_t+a_{i,t}$ with $\psi_t=t-T_0$ and $a_{i,t}\sim N(0,5)$, and $x_{i,t}^{(j)}=1+\bs\lambda_{0,i}'\bs F_{0,t}+\bs\lambda_{0,i}'(1,1)+\bs F_{0,t}'(1,1)+\eta_{i,t}^{(j)}$, with $\eta_{i,t}^{(j)}$ denoting independent standard normal noises for $j\in\{1,2\}$. Both $\bs{F}_{0,t}$ and $\zeta_t$ are independent standard normal, whereas both factor loadings and $\mu_i$ come from independent uniform distributions $U(-\sqrt{3},\sqrt{3})$. We assess the bias and standard error (in parenthesis) of the IFE estimators of $\bar\alpha$ for both specification, Static and Dynamic, and $\psi_{T_0+10}$ for Dynamic specification across 10,000 replications and different sample sizes.}}\\
    \toprule
        &&&& \multicolumn{2}{c}{Dynamic}  && Static\\
        \cline{5-6}\cline{8-8}
        $N_0$ & $N_1$ & $T_0$ & $T_1$ & $\bar\alpha$    & $\psi_{10}$       && $\bar\alpha$   \\
        \midrule
        40  & 20    & 15    & 10    & 0.1608 (1.360)    & 0.1373 (4.997)    && -0.5568 (2.703) \\
        200 & 20    & 15    & 10    & -0.0783 (0.758)   & -0.2429 (1.970)   && 1.0503 (0.371) \\
        500 & 20    & 500   & 10    & -0.0034 (0.074)   & -0.0020 (0.231)   && 0.0035 (0.074) \\
        500 & 500   & 500   & 100   & 0.0006 (0.007)    & 0.0023 (0.063)    && 0.9202 (0.351) \\
    \bottomrule
    \end{tabular}\label{tab:dynamic}
\end{table}